\documentclass[preprint,12pt]{elsarticle}

\usepackage{graphics}

\usepackage{times}
\usepackage{graphicx}
\usepackage{dsfont}
\usepackage{color}
\usepackage{subfigure}
\usepackage{enumerate}
\usepackage{tikz}
\usepackage{lscape}
\usepackage{rotating}

\usepackage{listings}
\usepackage{color}
 
\definecolor{dkgreen}{rgb}{0,0.6,0}
\definecolor{gray}{rgb}{0.5,0.5,0.5}
\definecolor{mauve}{rgb}{0.58,0,0.82}
 
\usepackage{amsmath,amsfonts,amssymb,amsthm,dsfont}
\usepackage[cp1250]{inputenc}

\usepackage{algorithmic}
\usepackage{algorithm}
\newcommand{\INDSTATE}[1][1]{\STATE\hspace{#1\algorithmicindent}}

\newtheorem{observation}{Observation}[section]
\newtheorem{theorem}{Theorem}[section]
\newtheorem{proposition}{Proposition}[section]
\newtheorem{lemma}{Lemma}[section]

\newtheorem{remark}{Remark}[section]
\newtheorem{example}{Example}[section]

\newtheorem*{SC}{Algorithm}
\newtheorem*{IRLST}{Theorem IRLS}
\newtheorem*{IRLSC}{Corollary IRLS}

\def\Z{\mathbb{Z}}
\def\R{\mathbb{R}}

\def\m{\mathrm{m}}

\def\1{\mathds{1}}

\def\for{\mbox{  for }}

\journal{Information Sciences}

\begin{document}

\begin{frontmatter}



\title{The memory centre}


\author{P. Spurek}
\address{
Faculty of Mathematics and Computer Science, 
Jagiellonian University, 
\L ojasiewicza 6, 
30-348 Krak\'ow, 
Poland}
\ead{przemyslaw.spurek@ii.uj.edu.pl}

\author{J. Tabor}
\address{Faculty of Mathematics and Computer Science, 
Jagiellonian University, 
\L ojasiewicza 6, 
30-348 Krak\'ow, 
Poland}
\ead{jacek.tabor@ii.uj.edu.pl}

\begin{abstract}
Let $x \in \R$ be given. As we know the, amount of bits needed to binary code $x$ with given accuracy ($h \in \R$) is  approximately 
$
\m_{h}(x) \approx \log_{2}(\max \{ 1 , | \frac{x}{h} |  \}). 
$
We consider the problem where we should translate the origin $a$ so that the mean amount of bits needed to code randomly chosen element from a realization of a random variable $X$ is minimal.
In other words, we want to find $a \in \R$ such that
$$
\R \ni a \to \mathrm{E} (\m_{h} (X-a))
$$
attains minimum.

We show that under reasonable assumptions, the choice of $a$ does not depend on $h$ asymptotically.
Consequently, we reduce the problem to finding minimum of the function
$$
\R \ni a \to \int_{\R}  \ln(|x-a|)f(x) dx,
$$
where $f$ is the density distribution of the random variable $X$.
Moreover, we provide constructive approach for determining $a$.

\end{abstract}


\end{frontmatter}



\section{Introduction}

Data compression is usually
achieved by assigning short descriptions (codes) to the most frequent outcomes
of the data source and necessarily longer descriptions to the less frequent
outcomes \cite{T,HHJ,SM} . 

For the convenience of the reader, we shortly present theoretical background of this approach.
Let $p=(p_0, \ldots, p_{n-1})$ be a probability distribution for a discrete random variable $X$. 
Assume that $l_i$ is the length of the code of $x_i$ for $i=0, \ldots, n-1$. 
Then the expected number of bits is given by
$
\sum_{i} p_i l_i.
$
The set of  possible codeword with uniquely decodable codes is limited by the Kraft inequality
$
\sum_{i}2^{-l_i} \leq 1.
$
It is enough to verify that lengths which minimize $\sum_{i} p_i l_i$ are given by $l_i=\log_2p_i$.
We obtain that minimal amount of information per one element in lossless coding is
Shannon entropy \cite{T} defined by
$$
H(X)=\sum - p_i \log_{2} p_i.
$$
By this approach various types of lossless data compression were constructed.  An optimal (shortest expected length) prefix code for a given distribution can be constructed by
a simple algorithm discovered by Huffman \cite{H}. 

If we want to consider continuous random variables 
and code with given maximal error $h$ we arrive at the notion of differential entropy \cite{T}. Let $f \colon \R \to \R$ be a continuous density distribution of the random variable $X$, and let $l \colon \R \to \R$ be the function of the code length. We divide the domain of $f$ into disjoint intervals of length $h$. Let $X_{h} := \frac{1}{h} \lfloor hx \rfloor$ be the discretization of $X$. The Shannon entropy of $X_h$ can be rewritten as follows 
\begin{equation}
\label{diff}
\begin{split}
H(X_{h}) & \approx \sum - f(x_i) h \log_{2}\left( f(x_i) h \right)\\
 & = \sum - f(x_i) \log_{2}\left( f(x_i) \right) h - \log_{2}\left( h \right) \sum f(x_i) h   \\
 & \approx \int  - f(x)\log_{2}\left( f(x) \right) dx -  \log_{2}\left( h \right) \int f(x)  dx \\
 & = \int - f(x_i)\log_{2}\left( f(x_i) \right) dx - \log_{2}\left( h \right).
\end{split}
\end{equation}

By taking the limit of $ H(X_h) + \log_2(h) $ as $h \to 0$, we obtain the definition of the differential entropy\footnote{Very often $\ln$ is used instead of $\log_2$, also in this article we use this convention.}
\begin{equation}
H(f) := - \int f(x) \log_{2} (f(x)) dx. 
\end{equation}

In this paper, we follow a different approach. Instead of looking for the best type of coding for a given dataset, 
we use standard binary coding\footnote{In the classical binary code we use one bit for the sign and then the standard binary representation. This code is not prefix so we have to mark ends of words. Similar coding is used in the decimal numeral system.} and we search for the optimal center $a$ of the coordinate system so that the mean amount of bits needed to code the dataset is minimal. 
The main advantage of this idea is that we do not have to fit the type of compression to a dataset.
Moreover, codes are given in very simple way.
This approach allows to immediately encrypt and decrypt large datasets (we use only one type of code).
Clearly, classical binary code is far from being optimal but it is simple and commonly used in practise.

The number of bits needed to code $x \in \Z$, by the classical binary code, is given by 
$$
\m(x) \approx  \log_{2}(\max \{ 1 , |x| \}) .
$$
Consequently, the memory needed to code $x \in \R$ 
with given accuracy $h$ is approximately 
$$
\m_{h}(x) = \m \left( \frac{x}{h} \right)  \approx  \log_{2} \left( \max \left\{ 1 , \left| \frac{x}{h} \right| \right\} \right) .
$$
Our aim is to find the place where to put the origin $a$ of the coordinate system so that the mean amount of bits needed to code randomly chosen an element from a sample from probability distribution of $X$ is minimal.
In other words, we want to find $a \in \R$ such that
$$
\mathrm{E} (\m_{h} (X-a)) = \int_{\R} \m_{h}(x-a)f(x)dx
$$
attains minimum, where $f$ is the density distribution of the random variable $X$.

Our paper is arranged as follows. In the next section we show that under reasonable assumptions,  the choice of $a$ does not depend on $h$ asymptotically. This reasoning
is similar to the derivation of the differential entropy $(\ref{diff})$.

In the third section, we consider the typical situation when the density distribution of a random variable $X$ is not known.
We use a standard kernel method to estimate the density $f$. 
Working with the estimation is possible, but from the numerical point of view, complicated and not effective. 
So in the next section we show reasonable approximation which has better properties. 

In the fourth section, we present our main algorithm and in Appendix B we put full implementation. 



In the last section, we present how our method works on typical datasets. 

\section{The kernel density estimation}

As it was mentioned in the previous section, our aim is to minimize, for a fixed $h$, the function $a \to \mathrm{E} (\m_{h} (X-a))$.
In this chapter, we show that the choice of $a$ (asymptotically) does not depend on $h$.
In Theorem \ref{theo:2.1}, we use a similar reasoning as in the derivation of differential entropy $(\ref{diff})$ and
we show that it is enough to consider the function
$$
M_f(a):= \int_{\R}  \ln(|x-a|)f(x) dx. 
$$

\begin{theorem}\label{theo:2.1}
We assume that the random variable $X$ has locally bounded density distribution $f \colon \R \to \R$. If $ M_f(a) < \infty $ for $a \in \R$, then
$$
\lim_{h \to 0} \left| \mathrm{E} (\m_{h} (X-a)) - \frac{1}{\ln(2)} M_f(a) + \log_{2}(h) \right|= 0 \quad  \for a \in \R .
$$
\end{theorem}
\begin{proof}
Let $a \in \R$ be fixed. Then
\begin{eqnarray*}
  \lefteqn{\mathrm{E} (\m_{h} (X-a)) = } \\
  & & \int_{\R}  \m_{h}(x-a)f(x) dx = \int_{\R}  \log_{2}\left( \max \left\{ 1 , \frac{|x-a|}{h} \right\} \right) f(x) dx = \\
  & & =\int_{\R \setminus (a-h , a + h) }  \log_{2}\left( \frac{|x-a|}{h} \right) f(x) dx =\\
  & & = \int_{\R \setminus (a-h , a + h) }  \log_{2} \left(|x-a|\right)f(x) dx + \int_{\R \setminus (a-h , a + h) } \log_{2} \left(h \right) f(x)dx = \\
  & & = \int_{\R}  \log_{2} \left(|x-a|\right)f(x) dx -
\int_{(a-h , a + h) }  \log_{2} \left(|x-a|\right)f(x) dx +\\
  & &  \int_{\R  } \log_{2} \left(h \right) f(x)dx -  
\int_{ (a-h , a + h) } \log_{2} \left( h \right) f(x)dx. 
\end{eqnarray*}
Since the function $f$ is a locally bounded density distribution so
$$
 - \int_{(a-h , a + h) }  \log_{2} \left(|x-a|\right)f(x) dx -  
\int_{ (a-h , a + h) } \log_{2} \left( h \right) f(x)dx  = 0.
$$
Consequently
$$
\lim_{ h \to 0} \left| \mathrm{E} (\m_{h} (X-a)) - \frac{1}{\ln(2)} M_f(a) + \log_{2}(h)  \right|= 0 \quad \for a \in \R .
$$
\end{proof}

As we see, when we increase the accuracy  ($h \to 0$) of coding the shape of the function $\mathrm{E} (\m_{h} (X-a))$
stabilizes (modulo subtraction of $\log_{2}(h)$). 

\begin{example}
Let $f$ be a uniform density distribution on interval $\left[ -\frac{1}{2},\frac{1}{2} \right]$.
Then 
$$
M_f(a)= \int_{\R}  \ln(|x-a|)f(x) dx 
=\int_{-\frac{1}{2}}^{\frac{1}{2}}  \ln(|x-a|) dx.
$$
Moreover, since 
$$
\int  \ln(x-a)  dx=
\ln(x-a)x-\ln(x-a)a-x-a,
$$
then we have
$$
M_f(a) = \left\{ \begin{array}{ll}
\ln \left( | \frac{1}{2} - a | \right) \left(\frac{1}{2} -a \right) + \ln\left( | \frac{1}{2}+a | \right) \left( \frac{1}{2}+a \right) -1& \textrm{for $|a| \neq \frac{1}{2}$,}\\
-1& \textrm{for $ |a| =\frac{1}{2}$.}
\end{array} \right.
$$
Function $M_f(a)$ is presented in Fig. \ref{fig:EX_2}.
\begin{figure}[htp]
  \begin{center}
\includegraphics[height=6cm]{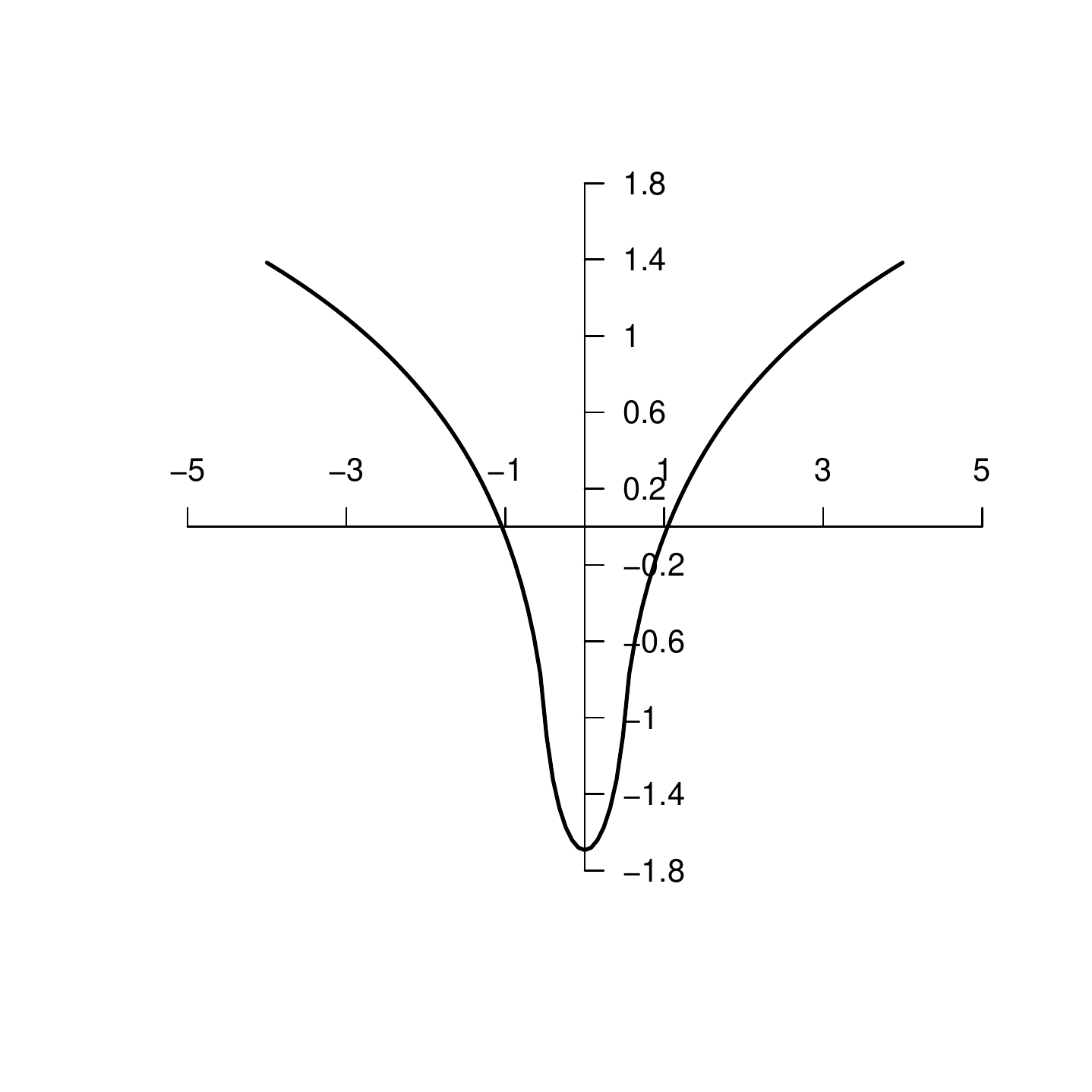}
  \end{center}
  \caption{Function $M_f(a)$ constructed for uniform density on $\left[ -\frac{1}{2} , \frac{1}{2} \right]$.}
  \label{fig:EX_2}
\end{figure}
It is easy to see that 
$
\min\left( M_f(a) \right) =0.
$
\end{example}

In a typical situation, we do not know the density distribution $f$ of random variable $X$. We only have a sample $ S=(x_1,\ldots,x_N)$ from $X$. To approximate $f$ we use kernel method \cite{BW}. 
For the convenience of the reader and to establish the notation we shortly present this method. The kernel  $K \colon \R \to \R$ is simply a function satisfying the following condition
$$
\int_{\R} K(x) dx= 1.
$$
Usually, $K$ is a symmetric probability density function. The kernel estimator of $ f$ with kernel $K$ is defined by
$$
\bar f(x) := \frac{1}{Nh} \sum_{i=1}^{N} K \left( \frac{x - x_{i} }{h} \right),
$$
where $h$ is the window width\footnote{Also called "the smoothing parameter" or "bandwidth" by some authors.}. 
Asymptotically optimal (for $N \to \infty$) choice of kernel K in class of symmetric and square-integrable functions is the Epanechnikov kernel \footnote{Often a rescaled (normalized) Epanechnikov kernel is used.}
$$
K(x) = \left\{ \begin{array}{lll}
\frac{3}{4}(1 - x^2) & \textrm{for $|x| < 1$,}\\
0 & \textrm{for $ |x| \ge 1 $.}
\end{array} \right.
$$
Asymptotically optimal choice of window width (under the assumption that density is Gaussian) is given by
$$
h \approx
2.35 s N ^{-\frac{1}{5}},
\quad \mbox{where } \quad
s = \sqrt{ \sum_{i=1}^{N} \frac{( x_{i} - m(S)  )^2}{N-1}}.
$$
Thus our aim is to minimize the function
$$
a \to M_S(a) := \frac{1}{Nh} \sum_{i=1}^{N}  \int_{x_i-h}^{x_i+h}   \ln(|x-a|) \frac{3}{4}\left( 1 -   \left( \frac{x - x_{i} }{h} \right)^2 \right) dx. 
$$
To compute $ M_{S}(a)$, we analyse the function $L \colon \R \to \R$ (see Fig. \ref{L}) given by:
$$
L \colon \R \ni a \to \frac{3}{4} \int_{-1}^{1}   \ln(|x-a|) \left( 1 -  x  ^2 \right) dx.
$$
\begin{figure}[htp]
  \begin{center}
\includegraphics[height=6cm]{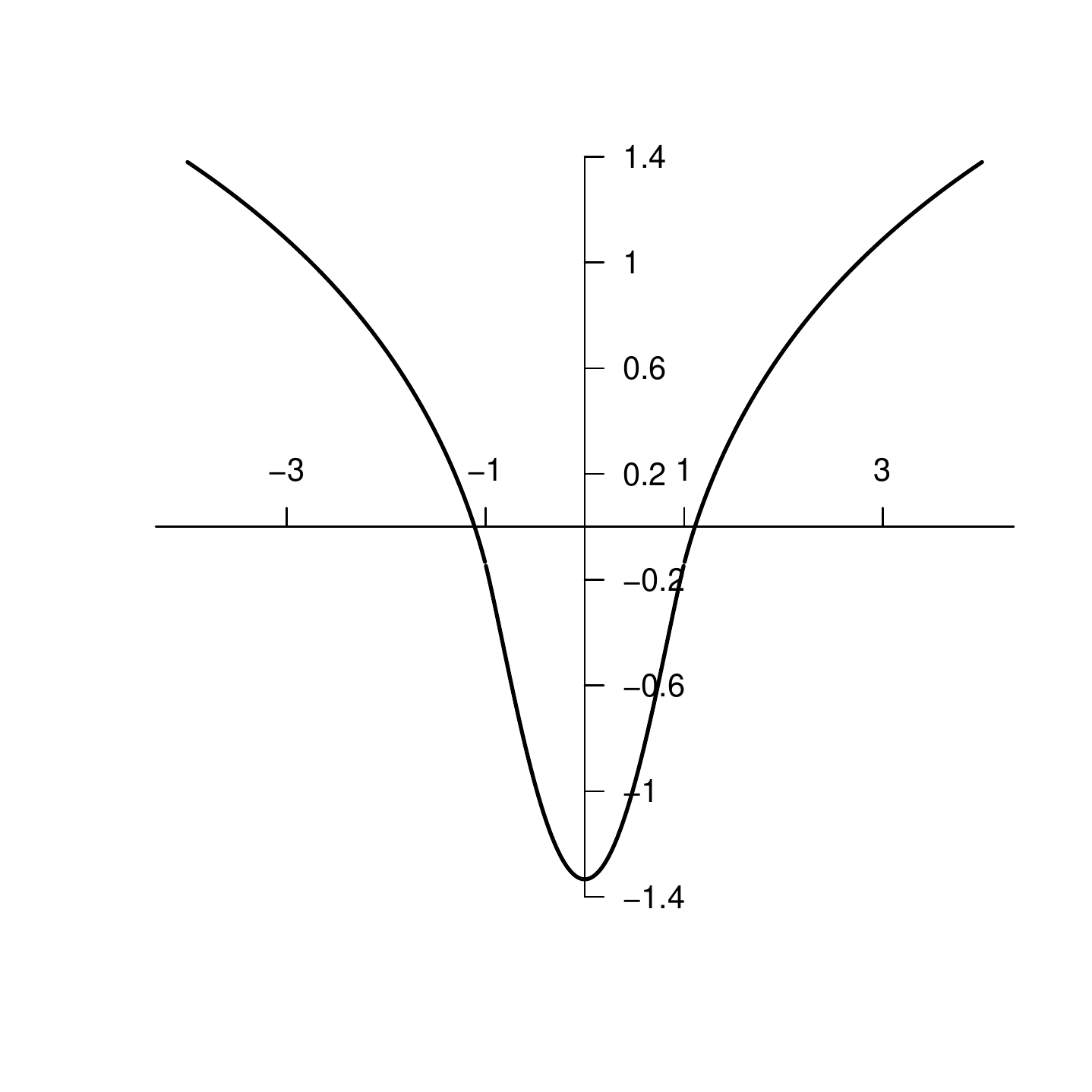}
  \end{center}
  \caption{Function $L$.}
  \label{L}
\end{figure}

\begin{lemma} \label{rem:2.1}
We have
$$
M_S(a) =  \ln(h) + \frac{1}{N} \sum_{i=1}^{N} L\left( \frac{x_i-a}{h} \right).
$$
\end{lemma}
\begin{proof}
By simple calculations we obtain
\begin{eqnarray*}
  \lefteqn{M_S(a) = } \\
  & & = \frac{1}{Nh} \sum_{i=1}^{N}  \int_{x_i-h}^{x_i+h}   \ln(|x-a|) \frac{3}{4}\left( 1 -   \left( \frac{x - x_{i} }{h} \right)^2 \right) dx= \\
  & & =\left| \begin{array}{c}
y = \frac{x-x_i}{h}\\
dy = \frac{dx}{h}\\
x= hy + x_i 
\end{array} \right|=
 \frac{1}{N} \sum_{i=1}^{N} \frac{3}{4} \int_{-1}^{1}   \ln\left(h \left| y + \frac{x_i - a}{h}\right|\right) \left( 1 -  y ^2 \right) dy = \\
  & & = \frac{1}{N} \sum_{i=1}^{N} \frac{3}{4} \int_{-1}^{1}   \ln\left(\left| y + \frac{x_i - a}{h}\right|\right) \left( 1 -  y ^2 \right) dy 
+ \ln(h)=
 \ln(h) + \frac{1}{N} \sum_{i=1}^{N} L\left( \frac{x_i-a}{h} \right).
\end{eqnarray*}
\end{proof}

\section{Approximation of the function $L$}

As it was shown in the previous section, the crucial role in our investigation is played by the function $L$
and therefore, in this chapter we study its basic properties.
Let us begin with the exact formula for $L$.

\begin{proposition}\label{fun:org}
The function $ L $ is given by the following formula
$$
L(a)=
 \left\{ \begin{array}{ll}
 \frac{1}{2}\ln(|1-a^2|)+ \left( \frac{3}{4}a-\frac{1}{4}a^3 \right)\ln\left( \left| \frac{1+a}{1-a} \right| \right) +\frac{1}{2}a^2  -\frac{4}{3}& \textrm{for $|a| \neq 1$,}\\
\ln(2)-\frac{5}{6} & \textrm{for $|a| =1$.}
\end{array} \right.
$$
Moreover, $L$ is even,
$L(0)=-\frac{4}{3}$ and
$ \lim \limits _{|a| \to \infty} ( L(a) - \ln(|a|) ) = 0$.

\end{proposition}
\begin{proof}
We consider the case when $a > 1$ (by similar calculation, we can get this result for all $a \in \R$)

\begin{eqnarray*}
  \lefteqn{ \int   \ln(x-a) \left( 1 -  x  ^2 \right) dx= } \\
  & & =
\left| \begin{array}{cc}
 \ln(x-a)&  1 -  x  ^2 \\
 \frac{1}{x-a} &  x - \frac{1}{3}x^3 
\end{array} \right|
= \left( x - \frac{1}{3}x^3 \right) \ln(|x-a|)  - \int  \frac{x - \frac{1}{3}x^3}{x-a} dx=\\
  & & =\ln(x-a) \left( - \frac{1}{3}x^3 + x -a + \frac{1}{3}a^3 \right) -x+a+\frac{1}{9}x^3+\frac{1}{6}x^2a+\frac{1}{3}xa^2-\frac{11}{18}a^3.
\end{eqnarray*}
Consequently
$$
L(a)=
 \left\{ \begin{array}{ll}
 \frac{1}{2}\ln(|1-a^2|)+ \left( \frac{3}{4}a-\frac{1}{4}a^3 \right)\ln\left( \left| \frac{1+a}{1-a} \right| \right) +\frac{1}{2}a^2  -\frac{4}{3}& \textrm{for $|a| \neq 1$,}\\
\ln(2)-\frac{5}{6} & \textrm{for $|a| =1$.}
\end{array} \right.
$$
As a simple corollary, we obtain that $L$ is even and
$L(0)=-\frac{4}{3}$.
To show the last property, we use the equality
$$
\frac{3}{4} \int_{-1}^{1}   \ln(|x-a|) \left( 1 -  x  ^2 \right) dx = \frac{3}{4} \int_{-1}^{1}   \ln(|x+a|) \left( 1 -  x  ^2 \right) dx.
$$
Then for $|a| > 1$, we obtain
$$
\frac{3}{4} \int_{-1}^{1}   \ln(|x-a|) \left( 1 -  x  ^2 \right) dx
=\frac{3}{8} \int_{-1}^{1}  \left( \ln(|x-a|)  +   \ln(|x+a|) \right) \left( 1 -  x^2 \right) dx =
$$
$$
= \frac{3}{8} \int_{-1}^{1}    \ln \left( a^2  \right) \left( 1 -  x  ^2 \right) dx +   \frac{3}{8} \int_{-1}^{1} \ln \left( 1 - \frac{x^2}{a^2} \right) \left( 1 -  x  ^2 \right) dx  =
$$
$$
= \frac{1}{2}\ln(a^2) +   \frac{3}{8} \int_{-1}^{1} \ln \left( 1 - \frac{x^2}{a^2} \right) \left( 1 -  x  ^2 \right) dx.
$$
Since
$$
\frac{3}{8} \int_{-1}^{1} \ln \left( 1 - \frac{x^2}{a^2} \right) \left( 1 -  x  ^2 \right) dx
\in \left[\min_{x \in [-1,1]}  \left( \frac{1}{2} \ln \left( 1 - \frac{x^2}{a^2} \right) \right) , \max_{x \in [-1,1]} \left( \frac{1}{2} \ln \left( 1 - \frac{x^2}{a^2} \right) \right) \right] =
$$
$$
=\left[ \left( \frac{1}{2} \ln \left( 1 - \frac{1}{a^2} \right) \right) , 0 \right],
$$
we get
$$
0 \geq \lim_{a \to \infty} (L(a) - \ln(a) ) \geq \lim_{a \to \infty} \left( \frac{1}{2} \ln \left( 1 - \frac{1}{a^2} \right) \right) = 0.
$$
\end{proof}

From the numerical point of view, the use of the function $L$ (Fig. \ref{L}) is complicated and not effective. The main problem is connected with a possible numerical instability for $a$ close to $1$. 
Moreover, in our algorithm we use the first derivative (more information in the next chapter and Appendix A) by considering the function
$$
\R_{+} \to L'(\sqrt{a}) =
 \left\{ \begin{array}{ll}
\frac{3}{8} \ln \left( \left| \frac{\sqrt{a}-1}{\sqrt{a}+1} \right| \right) a+\frac{1}{\sqrt{a}}\ln \left( \left| \frac{\sqrt{a}+1}{\sqrt{a}-1} \right| \right) + 2\sqrt{a} & \textrm{for $a \neq 1$,}\\
\frac{3}{4} & \textrm{for $a =1$.}
\end{array} \right.
$$
In this case, we have numerical instability for $a$ close to $0,1$.
Thus, instead of $L$, we use Cauchy M-estimator \cite{ZZ} $\bar L \colon \R \to \R$ which is given by 
$$
\bar L (a):= \frac{1}{2} \ln(e^{-\frac{8}{3}}+ a^2).
$$
The errors caused by the approximation are reasonably small in relation to those connected with kernell estimation\footnote{As it was said in this method one assume that dataset is realization of Gaussian random variable while usually, in practise,  does not have to.}.

\begin{observation}\label{fun:nowa}
The function $ \bar L $ is analytic, even, $\bar L(0)=L(0)$ and $ \lim \limits _{|a| \to \infty} ( \bar L (a) - L (a) ) = 0$.
\end{observation}

\begin{figure}[htp]
  \begin{center}
\includegraphics[height=6cm]{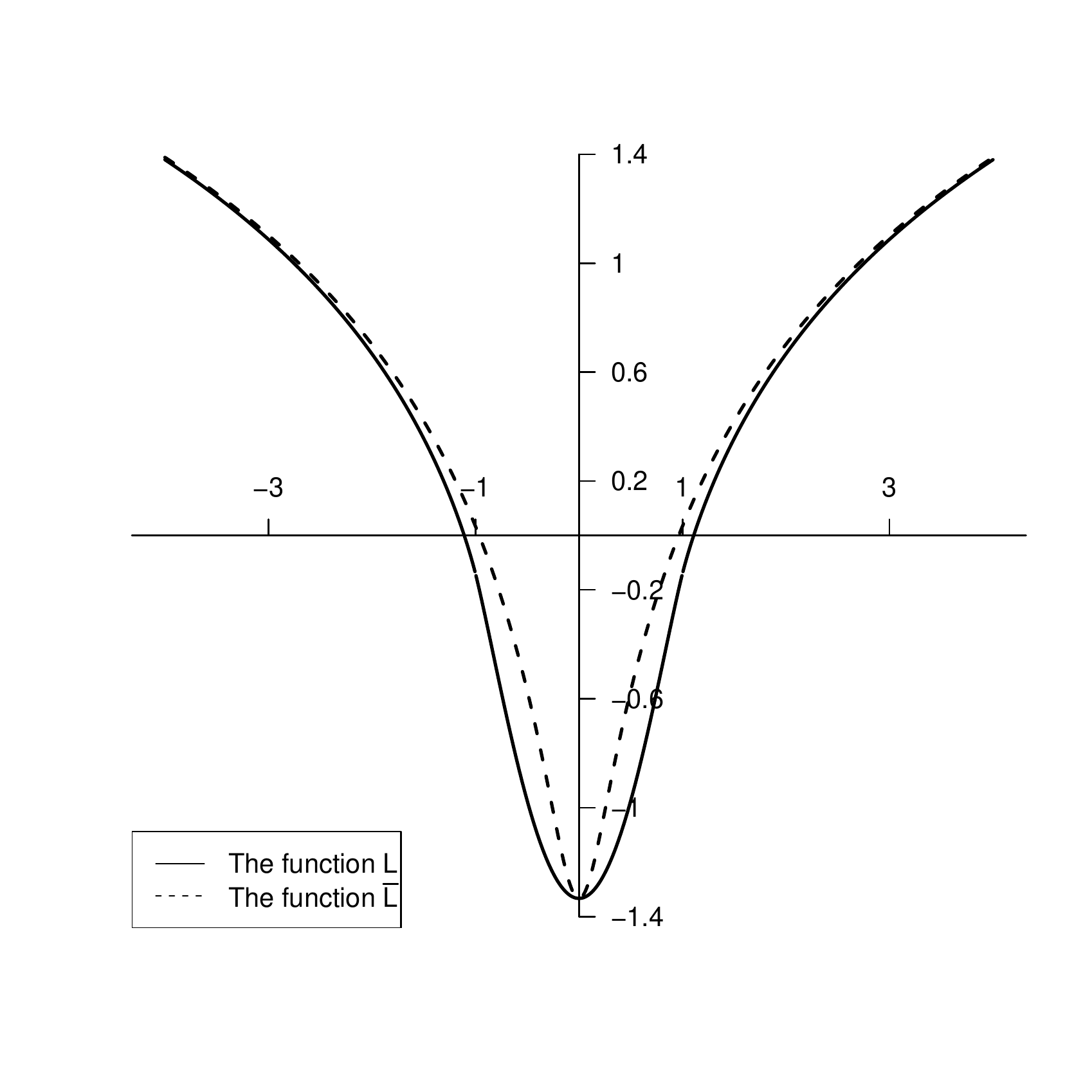}
  \end{center}
  \caption{Comparison of functions $L$ and $\bar L$.}
  \label{Ex_5}
\end{figure}

Consequently, the problem of finding the optimal (respectively to the needed memory) center of the coordinate system 
can be well approximated by searching for the global minimum of the function
$$
\bar M_S(a) := \frac{1}{N} \sum_{i=1}^{N}  
\bar L \left( \frac{a - x_{i} }{ h} \right) \for a \in \R. 
$$

\section{Search for the minimum}

In this section, we present a method of finding the minimum of the function $\bar M_S(a)$. 
We use the robust technique which is called M--estimator \cite{HR}. 

\begin{remark}
For the convince of the reader we shortly present the standard use of 
the M--estimator's method. 
Let $\{ x_1, \ldots, x_n \}$ be a given dataset. We  are looking for the best representation of the points
$$ 
\min \limits_a \sum_i ( x_i - a )^2.
$$
Normally we choose barycentre of the data but elements which are fare from the center usually courses undesirable effect.
The M--estimators try to reduce the effect of outliers by replacing the squares by another function (in our case Cauchy M-estimator \cite{ZZ})
$$
\min_a \sum_i L(x_i - a ),
$$
where $L$ is a symmetric, positive--definite function with a unique minimum at zero, and is
chosen to be less increasing than square. Instead of solving directly this problem, one usual
implements an Iterated Reweighted Least Squares method (IRLS) [see Appendix A]. 
In our case we are looking for
$$
\min_a \sum_{i} L(x_i - a ),
$$
where $L$ is interpreted as a function which describes the memory needed to code $x_i$ with respect to $a$.

\end{remark}

Our approach based on Iterated Reweighted Least Squares method (IRLS), is  similar to that presented in \cite{W} and \cite{I}. For convenience of the reader, we include the basic theory connected with this method in Appendix A. 

In our investigations the crucial role is played by the following proposition.

\begin{IRLSC}[see Appendix A]\label{prop:6.1}
Let $f:\R_+ \to \R_+$, $f(0)=0$ be a given concave and differentiable function
and let $S = \{ x_i, \ldots, x_N\}$ be a given data-set.
We consider the function
$$
F(a):=\sum_i f(|a-x_i|^2) \quad \for a \in \R.
$$
Let $\bar a \in \R$ be arbitrarily fixed and let
$$
w_{i} = f'(|x_i-\bar a|^2) \quad \for i = 1 \ldots N \quad \mbox{ and } \quad
\bar  a_w = \frac{1}{\sum _{i=0}^{N} w_{i} } \sum \limits_{i=0}^{N}  w_{i}  x_{i}. 
$$
Then
$$
F( \bar a_w  ) \leq F \left( \bar a \right).
$$
\end{IRLSC}

Making use of the above corollary, by substitution $\bar a \to \bar a_w$ in each step we come closer to a local minimum of the function $F$. 
It is easy to see, that $\bar L$ defined in the previous section, satisfies the assumptions of Corollary IRLS. 
Let $S = (x_1 , \ldots, x_N)$ be a given realization of the random variable X and let
$$
h = 2.35  N ^{-\frac{1}{5}}\sqrt{ \sum_{i=1}^{N} \frac{( x_{i} - m(S)  )^2}{N-1}}.
$$
The algorithm (based on IRLS) can be described as follows.
\begin{SC}
\begin{center}
	\begin{minipage}[t]{0.9\textwidth}
		\begin{algorithmic}
			\STATE {\bf initialization}
			\STATE \textit{stop condition}
			\INDSTATE  $\varepsilon > 0$
			\STATE \textit{initial condition}
			\INDSTATE  $j=0$
			\INDSTATE  $a_j=m(S) $
			\REPEAT 
			\STATE \textit{calculate} 
			\INDSTATE  $w_{i}:=\left( e^{-\frac{8}{3}}+ \frac{1}{h^2} (x_i-a_j)^2 \right)^{-1} \for i = 1, \ldots, N$ 
			\INDSTATE  $j=j+1$
			\INDSTATE  $ a_j =  \frac{1}{\sum _{i=0}^{N} w_{i} } \sum \limits_{i=0}^{N}  w_{i}  x_{i} $
			\UNTIL{ $| a_j - a_{j-1} | < \varepsilon$}
		\end{algorithmic}
	\end{minipage}
\end{center}
\end{SC}
\begin{figure}[htp]
  \begin{center}

\includegraphics[height=6cm]{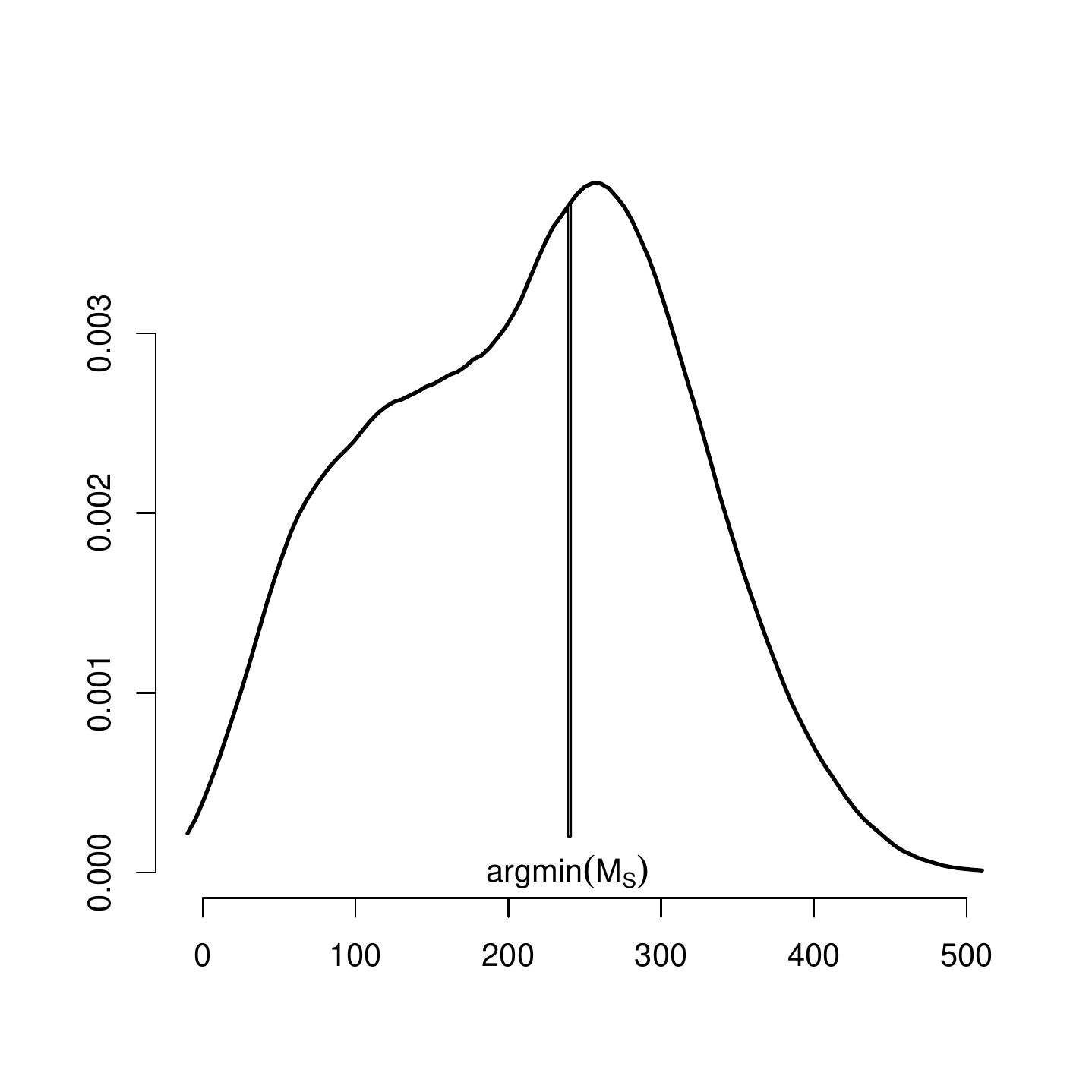}
  \end{center}
  \caption{Estimation of the density distribution for the Forex data.}
  \label{Forex_d}
\end{figure}
The first initial point can be chosen in many different ways. We usually start from the barycentre of the dataset because, for equal weights, the barycentre is the best approximation of the sample (for full code written in R Project, see Appendix B).

Now we show how our method works in practise. 
\begin{example}
Let $S$ be the sample of the index of USD/EUR from Forex stoke \cite{www}.
The density obtain by the kernel method is presented at Fig. \ref{Forex_d}.
As a result of the algorithm we obtained $\mathrm{alg\_centre} (S) = 238.4174$. We compare our result with the global minimum $\mathrm{argmin}  (M_{S}) =239.509 $ and the barycentre of data
$ m(S) = 212.8004 $:
$$
\min M_{S} = 4.0477, 
$$
$$
M_{S} ( \mathrm{alg\_centre} (S)  )= 4.0478,
$$
$$
M_{S} ( m(S) )= 4.0768.
$$

As we see the difference between $\min M_{S}$ and $M_{S} ( \mathrm{alg\_centre} (S)  )$ is small. 
Moreover, the barycentre gives a good approximation of the memory centre but, as we see in next examples, 
the difference can be large for not uni--modal densities.
\end{example}

In the next step we consider a random variables of the form
$$
X := p_1 \cdot X_1 + p_2 \cdot X_2
$$
where $X_1, X_2$ are two normal random variables\footnote{The normal random variables with means $m$ and  the standard deviation $s$ we denote by $N_{( m , s )}$.} or two uniform random variables\footnote{The uniform random variable on the interval $[a,b]$ we denote by  $U_{[ a , b ]}$.}. 



In Table \ref{tab:1}, we present comparison of the result of our algorithm and global minimum of the function $M_S$ where $S$ is the realization of random variable $X$ of size 500 with different parameters.
As we see, in the second and the third columns the algorithm which uses the function $\bar L$, gives a good approximation of the minimum for the function $L$.
It means that the use of $\bar L$ from the third section is reasonable and causes minimal errors.

In our cases, we obtained a good approximation of the global minimum.  
Moreover, the difference between the minimum of the original function and the result of our algorithm is small (see fifth column). Consequently, we see that the barycentre of data sets is a good candidate for initial point.

Clearly (see sixth column), we see that the barycentre of a data is not a good approximation 
of the memory center, especially in the situation of not uni-modal densities. 

\newenvironment{changemargin}[2]{%
  \begin{list}{}{%
    \setlength{\topsep}{0cm}%
    \setlength{\leftmargin}{-0.8cm}%
    \setlength{\rightmargin}{-1cm}%
  }%
  \item[]}{\end{list}}

\begin{table}[t]

\begin{changemargin}

	\begin{tabular}{| l | l | l | l | l |  }

    \hline
     Model &  &  & &\\ 

 $ p_1 \cdot X_1 + p_2 \cdot X_2	$ &  $a_{ r}$  & $a_m$ & $ M_S(a_r) - M_S(a_m) $ & $ M_S( m(S) ) - M_S(a_m) $ \\ \hline 

$ 0.6  N_{( -1 , 1 )} +  0.4  N_{( 1 , 1 )} $&  -0.347  &  -0.379  &  0.00017  &  0.00658  \\ \hline
$ 0.4  N_{( -6 , 1 )} +  0.6  N_{( 6 , 1 )} $&  \ 5.803  & \ 5.676  &  0.00079  &  0.55374  \\ \hline
$ 0.4  N_{( -1 , 1 )} +  0.6  N_{( 1 , 1 )} $&  \ 0.483  & \ 0.416  &  0.00093  &  0.01353  \\ \hline
$ 0.3  N_{( -6 , 0.5 )} +  0.7  N_{( 6 , 1 )} $& \ 5.979  &\  5.863  &  0.00089  &  0.57296  \\ \hline
$ 0.2  N_{( -2 , 0.5 )} +  0.8  N_{( 3 , 2 )} $& \ 2.721  &\  2.770  &  0.00021  &  0.05129  \\ \hline
$ 0.6  U_{[ -3 , -1 ]} +  0.4  U_{[ 0 , 1 ]} $&  -1.805  &  -1.824  &  0.00014  &  0.19388  \\ \hline
$ 0.4  U_{[ -3 , -1 ]} +  0.6  U_{[ 0 , 1 ]} $& \ 0.364  &\  0.403  &  0.00144  &  0.41139  \\ \hline
$ 0.3  U_{[ -3 , -1 ]} +  0.7  U_{[ 0 , 1 ]} $& \ 0.433  &\  0.447  &  0.00028  &  0.44633  \\ \hline
$ 0.2  U_{[ -2 , -1 ]} +  0.8  U_{[ 1 , 2 ]} $& \ 1.403  &\  1.445  &  0.00265  &  0.38220  \\ \hline
$ 0.2  U_{[ -5 , -2 ]} +  0.8  U_{[ 3 , 4 ]} $& \ 3.464  &\  3.439  &  0.00019  &  0.50817  \\ \hline

	\end{tabular}

\end{changemargin}

\caption{In this table we have following the notation: $a_r = \mathrm{alg\_centre}(S) $ and $a_m = \mathrm{argmin}(M_S)$.}
\label{tab:1}

\end{table}

\section{Appendix A}

In the fourth section, we presented a method for finding the minimum of the function $\bar M_S(a)$. Our approach based on the Iterated Reweighted Least Squares algorithm (IRLS). The method can be applied in  statistic and computer since. We have used them to minimize the function $\bar M_S(a)$. Similar approach is presented in \cite{W} and \cite{I}. For convenience of the reader, we show the basic theoretical information about IRLS algorithm. The main theorem, related with this method, can be formulated as follows: 

\begin{IRLST}\label{IRLS}
Let $f_i:\R_+ \to \R_+$, $f(0)=0$ be a set of concave and differentiable function
and let $X=(x_i)$ be a given data-set.
Let $\bar a \in \R$ be fixed.
We consider functions
$$
F(a):=\sum_i f_i(|a-x_i|^2) \quad \for a \in \R
$$
and
$$
H(a):=\sum_i [f_i(|\bar a-x_i|^2)-f_i'(|\bar a-x_i|^2)|\bar a-x_i|^2]
+f_i'(|\bar a-x_i|^2)|a-x_i|^2 \quad \for a \in \R.
$$
Then $ H(\bar a)=F(\bar a)$ and
$$
H \geq F.
$$

\end{IRLST}

\begin{proof}
Let $i$ be fixed. For simplicity, we denote $f = f_i$.

Let us first observe that, without loss of generality, we may assume that
$\bar x=0$ (we can make an obvious substitution $a \to a+\bar x$).

Then since all the considered functions are radial, it is sufficient
to consider the one dimensional case when $N=1$. Thus, from now on we assume that $N=1$ and $\bar x=0$. Since the functions are even, it is sufficient to consider the situation on $\R_+$.

Concluding: we are given $\bar a \in \R_+$ and consider functions
$$
g:\R_+ \ni a \to f(a^2)
$$
and
$$
h:\R_+ \ni a \to [f(\bar a^2)-f'(\bar a^2)\bar a^2]
+f'(\bar a^2)a^2.
$$
Clearly, $h(\bar a)=g(\bar a)$. We have to show that $h \geq g $.
We consider two cases.

Let us first discuss the situation on the interval $[0,\bar a]$.
We show that $g-h$ is increasing on this interval, since coincide at $\bar a$
this makes the proof completed. Clearly $g$ and $h$ are absolutely continuous functions
(since $f$ is concave). Thus, to prove that $g-h$ is increasing on $[0,\bar a]$,
it is sufficient to show that $g' \geq h'$ a.e. on $[0,\bar a]$. But $f$ is
concave, and therefore $f'$ is decreasing, which implies that
$$
g'(a)=f'(a^2)2a \geq f'(\bar a^2)2a=h'(a).
$$

So let us consider the situation on the interval $[\bar a,\infty)$.
We will show that $g-h$ is decreasing on $[\bar a,\infty)$. Since $(g-h)(\bar a)=0$,
this is enough. Note that
$$
(g-h)'(a)=f'(a^2)2a-f'(\bar a^2)2a=2a(f'(a^2)-f'(\bar a^2)) \leq 0.
$$

Now, the assertion of the theorem is a simple consequence of the previous property.
\end{proof}



Now we can form the most important results:

\begin{IRLSC}\label{prop:7.1}
Let $f:\R_+ \to \R_+$, $f(0)=0$ be a given concave and differentiable function
and let $S = \{ x_i \}_{i=1}^{N}$ be a given data-set.
We consider the function
$$
F(a)=\sum_i f(|a-x_i|^2) \quad \for a \in \R.
$$
Let $\bar a \in \R$ be arbitrarily fixed and let
$$
w_{i} = f'(|\bar a - x_i|^2) \quad \for i = 1 \ldots N \quad \mbox{ and } \quad
a_w = \frac{1}{\sum _{i=0}^{N} w_{i} } \sum \limits_{i=0}^{N}  w_{i}  x_{i}. 
$$
Then
$$
F( a_w  ) \leq F \left( \bar a \right).
$$
\end{IRLSC}
\begin{proof}
Let $H$ be defined like in Theorem IRLS
$$
H(a)=\sum_i [f(|\bar a-x_i|^2)-f'(|\bar a-x_i|^2)|\bar a-x_i|^2]
+f'(|\bar a-x_i|^2)|a-x_i|^2 \for a \in \R.
$$
Moreover by Theorem IRLS we have 
$$
F(\bar a) = H(\bar a).
$$
Function $H$ is quadratic so the minimum is
$$
a_w = \frac{1}{\sum _{i=0}^{N} w_{i} } \sum \limits_{i=0}^{N}  w_{i}  x_{i} 
$$
and consequently
$$
F(\bar a) = H(\bar a) \geq H(a_w).
$$
\end{proof}
Making use of the above theorem, we obtain a simple method of finding a better approximation of the minimum. For given $a \in \R $, by taking weighted average $a_w$ (see Corollary IRLS) we find the point which
reduces the value of the function.
So, to find minimum, we iteratively calculate weighted barycentre of data set.

\section{Appendix B}

In this section we present source code of our algorithm written in R Project.

\begin{lstlisting}[language=R]
#The derivative of the function bar L
d_bar_L = function(s){(exp(-8/3)+s)^(-1)}

#The function which describe a single iteration
#S - data p - starting point 
iteration = function(p,S){
  suma=0
  weight=0
  v=sd(S) #standard deviation
  h=2.35*v*(length(S))^(-1/5)
  for(s in S){
    weight_s=d_bar_L((abs(s-p)^2)/(h^2))
    suma=suma+weight_s*s
    weight=weight+weight_s
  }
  suma/weight
}

#The function which determine a local minimum
#S - data 
#e - epsilon
#N - maximal number of iterations
look_minimum=function(S,e,N){
  ans=mean(S)
  dist=1
  while(dist>e & N>0){ 
    ans_n=iteration(ans,S)
	N=N-1
	dist=abs(ans_n-ans)
	ans=ans_n
  }
  ans
}
\end{lstlisting}

In our simulation we use $\varepsilon = 0.001$ and  $N = 50$.





\bibliographystyle{model1-num-names}
\bibliography{<your-bib-database>}







\end{document}